\documentclass[11pt]{article}
\usepackage{fullpage,epsf,amssymb,amsthm,amsfonts,amsmath,latexsym, graphicx,array,extarrows,mathtools,mathstyle,mathrsfs,slashed}
\usepackage[font=footnotesize,labelfont=bf]{caption}
\usepackage[normalem]{ulem}
\usepackage{float}

\newcommand{\overbar}[1]{\mkern 1.5mu\overline{\mkern-1.5mu#1\mkern-1.5mu}\mkern 1.5mu}

\let\oldsqrt\sqrt

\def\sqrt{\mathpalette\DHLhksqrt}
\def\DHLhksqrt#1#2{
\setbox0=\hbox{$#1\oldsqrt{#2\,}$}\dimen0=\ht0
\advance\dimen0-0.2\ht0
\setbox2=\hbox{\vrule height\ht0 depth -\dimen0}
{\box0\lower0.4pt\box2}}

\title{\bf{Conditions on the violation of the cluster decomposition property in QCD}} 
\author{\large{Peter Lowdon} \\
\ \\
\textit{\small{Physik-Institut, Universit\"at Z\"urich, Winterthurerstrasse 190, 8057 Z\"urich, Switzerland}} \\
\textit{\small{E-mail: lowdon@physik.uzh.ch}}}
\date{}
\begin{document}
\begin{flushright} ZU--TH 35/15   \end{flushright}
\vspace{15mm} 
{\let\newpage\relax\maketitle}
\setcounter{page}{1}
\pagestyle{plain}

\maketitle
\setcounter{page}{1}
\pagestyle{plain}
 
\abstract \noindent 
The behaviour of correlators at large distances plays an important role in the dynamics of quantum field theories. In many instances, correlators satisfy the so-called \textit{cluster decomposition property} (CDP), which means that they tend to zero for space-like asymptotic distances. However, under certain conditions it is possible for correlators to violate this property. In the context of quantum chromodynamics (QCD), a violation of the CDP for correlators of clusters involving coloured fields implies that the strength of the correlations between the coloured degrees of freedom in these clusters increases at large distances, which is a sufficient condition for confinement. In this paper we establish a criterion for when the CDP is violated. By applying this criterion to QCD, it turns out that certain lattice results involving the quark and gluon propagators can be interpreted as evidence that quarks and gluons are confined due to a violation of the CDP.

\pagebreak
\newtheorem{mydef}{Definition}
\newtheorem{theorem}{Theorem}
\newtheorem{prop}{Proposition}
\newtheorem{lemma}{Lemma}
\newtheorem{corr}{Corollary}
\newcommand{\lpipe}{\rule[-0.4ex]{0.90pt}{2.3ex}}
\setlength\fboxrule{0.5mm}
\setlength\fboxsep{4mm}

\newtheorem*{theorem*}{Theorem}

\section{Introduction}
\label{paper4_intro}

Establishing the structure of correlators in a quantum field theory (QFT) is central to understanding the characteristics of the theory. In particular, the space-like asymptotic behaviour of truncated correlators comprised of field clusters determines how the strength of the correlations between the field degrees of freedom in these clusters changes as the distance between the clusters grows~\cite{Strocchi76,Strocchi78}. This behaviour has been investigated several times in the literature\footnote{See~\cite{Haag58,Araki60,Ruelle62,Araki_Hepp_Ruelle62} and references within.}, and a result of particular importance is the \textit{cluster decomposition theorem}~\cite{Strocchi76,Araki_Hepp_Ruelle62}. This theorem characterises the asymptotic correlations between field clusters in QFTs which satisfy certain general axioms\footnote{See~\cite{Streater_Wightman64,Nakanishi_Ojima90,Haag96} for a further discussion of these axioms and their physical motivation.}. For QFTs which have a space of states with a positive-definite inner product, the cluster decomposition theorem implies that truncated correlators of field clusters tend to zero for large space-like distances~\cite{Araki_Hepp_Ruelle62}, and in the case where the theory has a mass gap $(0,M)$, the rate of vanishing is faster than any inverse power of the distance. Correlators which have this behaviour are said to preserve the \textit{cluster decomposition property} (CDP). Physically, this means that the correlation strength between field clusters always decreases asymptotically if the corresponding correlator preserves the CDP. However, in the case of quantised gauge theories such as QCD, the standard QFT axioms no longer apply because locality is lost due to the gauge symmetry of the theory. Nevertheless, one can restore locality by adopting a \textit{local quantisation}\footnote{One such example of this is the \textit{BRST quantisation} of Yang-Mills theories~\cite{Strocchi13}.}, which requires that the inner product in the space of states $\mathcal{V}$ is no longer positive-definite~\cite{Strocchi13}. Having a space of states $\mathcal{V}$ with an indefinite inner product implies many consequences, including the modification of the cluster decomposition theorem. This modification has the particularly important feature that correlators are permitted to violate the CDP~\cite{Strocchi76}. \\

\noindent
In the context of QCD, the modified cluster decomposition theorem has an important physical application -- it provides a mechanism by which coloured degrees of freedom can be confined. Specifically, if the truncated correlator of coloured field clusters violates the CDP, then the correlations between the coloured degrees of freedom in these clusters are not damped, no matter how far they are separated. Thus, the measurement of a state associated with one of the coloured fields cannot be performed independently of the other, and hence the detection of individual coloured states is not possible, which is a sufficient condition for confinement~\cite{Nakanishi_Ojima90,Roberts_Williams_Krein91}. The failure of the CDP for such correlators is related to the inability to construct physical asymptotic states associated with the coloured fields~\cite{Strocchi78}. Obtaining a better understanding of the general conditions under which the CDP is violated is therefore clearly important if one wants to establish whether QCD confinement occurs in this manner. It is these issues which we aim to shed light on in this paper. \\

\noindent
The remainder of this paper is structured as follows: in Sec.~\ref{paper4_section1} the cluster decomposition theorem is discussed, and conditions concerning the violation of the CDP are derived; in Sec.~\ref{paper4_section2} the general Lorentz structure of QFT correlators is outlined and related to the CDP conditions in Sec.~\ref{paper4_section1}; in Sec.~\ref{paper4_section3} the results of Secs.~\ref{paper4_section1} and~\ref{paper4_section2} are applied to QCD, and in particular specific lattice calculations involving the quark and gluon propagators; and finally in Sec.~\ref{paper4_concl} the findings of the paper are summarised.

\section{The Cluster Decomposition Theorem}
\label{paper4_section1}

In locally quantised QFTs, a local field algebra is maintained at the expense of adding additional degrees of freedom to the theory, and results in a space of states $\mathcal{V}$ with an indefinite inner product~\cite{Strocchi13}. The physical states $\mathcal{V}_{\text{phys}} \subset \mathcal{V}$ are subsequently characterised by a so-called \textit{subsidiary condition}. For Yang-Mills gauge theories such as QCD, the most common local quantisation approach is \textit{BRST quantisation}. In this case gauge fixing and ghost terms are added to the Lagrangian density, and the corresponding subsidiary condition is: $Q_{B}\mathcal{V}_{\text{phys}}=0$, where $Q_{B}$ is the conserved charge associated with the residual BRST symmetry which the modified Lagrangian density possesses. The Hilbert space is defined by $\mathcal{H}:=\overbar{\mathcal{V}_{\text{phys}}\slash \mathcal{V}_{0}}$, where $\mathcal{V}_{0} \subset \mathcal{V}_{\text{phys}}$ contains the zero norm states, and the closure indicates that particular limit states are also in $\mathcal{H}$~\cite{Nakanishi_Ojima90}. As mentioned in Sec.~\ref{paper4_intro}, if a QFT has been locally quantised then the space-like asymptotic behaviour of cluster correlators is modified with respect to non-locally quantised theories. In particular, one has the following theorem~\cite{Strocchi76}: 
\ \\
\begin{theorem}[Cluster Decomposition]
\begin{align*}
\left|\langle 0 | \mathcal{B}_{1}(x_{1})\mathcal{B}_{2}(x_{2}) |0\rangle^{T} \right| \leq \left\{
     \begin{array}{ll}
      C_{1,2}[\xi]^{2N-\frac{3}{2}} e^{-M [\xi]}\left(1 + \frac{|\xi_{0}|}{[\xi]}\right) , & \text{with a mass gap $(0,M)$ in $\mathcal{V}$} \\
        \widetilde{C}_{1,2} [\xi]^{2N-2}\left(1 + \frac{|\xi_{0}|}{[\xi]^{2}}\right), & \text{without a mass gap in $\mathcal{V}$}
     \end{array}
     \right.
\end{align*}
where: $ \langle 0 | \mathcal{B}_{1}(x_{1})\mathcal{B}_{2}(x_{2}) |0\rangle^{T} = \langle 0 | \mathcal{B}_{1}(x_{1})\mathcal{B}_{2}(x_{2}) |0\rangle - \langle 0 | \mathcal{B}_{1}(x_{1})|0\rangle \langle 0 | \mathcal{B}_{2}(x_{2}) |0\rangle$, $N\in \mathbb{Z}_{\ge 0}$, $\xi = x_{1}-x_{2}$ is large and space-like, and $C_{1,2},\widetilde{C}_{1,2}$ are constants independent of $\xi$ and $N$. 
\label{cdt}
\end{theorem}
\ \\
\noindent
The \textit{cluster correlator} $\langle 0 | \mathcal{B}_{1}(x_{1})\mathcal{B}_{2}(x_{2}) |0\rangle$ is defined by:
\begin{align}
\langle 0 | \mathcal{B}_{1}(x_{1})\mathcal{B}_{2}(x_{2}) |0\rangle &:= \langle 0 | \phi_{1}(f_{x_{1}}^{(1)})\phi_{2}(f_{x_{2}}^{(2)}) |0\rangle \nonumber \\
&= \int d^{4}y_{1} d^{4}y_{2} \, \langle 0 | \phi_{1}(y_{1})\phi_{2}(y_{2}) |0\rangle f^{(1)}(y_{1}-x_{1}) f^{(2)}(y_{2}-x_{2})
\label{corr_smear}
\end{align} 
where $f^{(i)} \in \mathcal{D}(\mathbb{R}^{1,3})$, and $\phi_{k}$ are the quantised \textit{basic fields} in the theory\footnote{The so-called \textit{basic fields}~\cite{Haag58,Araki_Hepp_Ruelle62} consist of polynomials of the simplest fields, such as the quark $\psi$ and gluon $A_{\mu}^{a}$ fields in QCD.}. The test functions $f^{(i)}$ are chosen to have compact support because this allows the operators $\phi_{1}(f_{x_{1}}^{(1)})$ and $\phi_{2}(f_{x_{2}}^{(1)})$ to be interpreted as clusters containing the field degrees of freedom of $\phi_{1}$ and $\phi_{2}$, centred around the points $x_{1}$ and $x_{2}$ respectively. The precise definition of $[\xi]$ is outlined in~\cite{Araki_Hepp_Ruelle62}, but for large space-like distances $[\xi]$ can be approximated by $|{\boldsymbol{\xi}}|:= r$~\cite{Nakanishi_Ojima90}. It follows from Theorem~\ref{cdt} that in a locally quantised QFT such as QCD, if $\mathcal{V}$ does not have a mass gap, then the correlation strength $F^{\phi_{1}\phi_{2}}(r)$ between the clusters of the fields $\phi_{1}$ and $\phi_{2}$ has the asymptotic behaviour:  
\begin{align}
F^{\phi_{1}\phi_{2}}(r) \sim r^{2N-2}, \hspace{3mm} \text{for $r \rightarrow \infty$}
\label{F_corr}
\end{align}
This has the important consequence that the CDP can be violated for $N \neq 0$, which as discussed in Sec.~\ref{paper4_intro}, is particularly relevant in the context of confinement in QCD. \\

\noindent
In axiomatic formulations of QFT~\cite{Streater_Wightman64}, the basic field correlators $\langle 0 | \phi_{1}(y_{1})\phi_{2}(y_{2}) |0\rangle = T_{(1,2)}(y_{1}-y_{2})$ are defined to be tempered distributions $\mathcal{S}'(\mathbb{R}^{1,3})$, and hence their Fourier transforms $\widehat{T}_{(1,2)}(p) = \mathcal{F}\left[T_{(1,2)}(y_{1}-y_{2})\right]$ are also in $\mathcal{S}'(\mathbb{R}^{1,3})$. Moreover, since 
\begin{align*}
\langle 0 | \mathcal{B}_{1}(x_{1})\mathcal{B}_{2}(x_{2}) |0\rangle^{T} = \mathcal{T}_{(1,2)}^{T}(x_{1}-x_{2})
\end{align*}
is the convolution of $\langle 0 | \phi_{1}(y_{1})\phi_{2}(y_{2}) |0\rangle^{T}= T_{(1,2)}^{T}(y_{1}-y_{2})$ with Schwartz test functions $f^{(i)} \in \mathcal{D}(\mathbb{R}^{1,3}) \subset \mathcal{S}(\mathbb{R}^{1,3})$, it follows that both $\mathcal{T}_{(1,2)}^{T}$ and its Fourier transform $\widehat{\mathcal{T}}_{(1,2)}^{T}(p)$ are tempered distributions~\cite{Strichartz94}. Due to the \textit{spectral condition} axiom\footnote{The \textit{spectral condition} axiom in QFT is related to the physical assumption that the energy spectrum of the theory is bounded from below~\cite{Strocchi13}.}, $\widehat{T}_{(1,2)}^{T}(p)$ is also defined to have support in the closed forward light cone $\overbar{V}^{+}$. By Eq.~(\ref{corr_smear}), $\widehat{\mathcal{T}}_{(1,2)}^{T}(p)$ can be written:
\begin{align}
\widehat{\mathcal{T}}_{(1,2)}^{T}(p) = \hat{f}^{(1)}(-p)\hat{f}^{(2)}(p) \widehat{T}_{(1,2)}^{T}(p) = g(p) \widehat{T}_{(1,2)}^{T}(p)
\label{smear_p}
\end{align} 
where $\hat{f}^{(i)}= \mathcal{F}\left[f^{(i)}\right] \in \mathcal{S}(\mathbb{R}^{1,3})$, and $g\in \mathcal{S}(\mathbb{R}^{1,3})$. So the convolution of $T_{(1,2)}^{T}$ with test functions in position space becomes a multiplication of $\widehat{T}_{(1,2)}^{T}$ with test functions in momentum space. Since $\text{supp} \, \widehat{T}_{(1,2)}^{T} \subset \overbar{V}^{+}$, Eq.~(\ref{smear_p}) implies that $\widehat{\mathcal{T}}_{(1,2)}^{T}(p)$ must also have support in $\overbar{V}^{+}$. The parameter $N$ in Theorem~\ref{cdt} depends on the structure of $\widehat{\mathcal{T}}_{(1,2)}^{T}(p)$. In particular, $N$ is characterised by the following general theorem~\cite{Bros_Epstein_Glaser67}:
\ \\
\begin{theorem}[Bros-Epstein-Glaser] 
Let $\mathscr{T} \in \mathcal{S}'(\mathbb{R}^{1,3})$ be a tempered distribution with support in $\overbar{V}^{+}$. Then there exists a non-negative integer $N \in \mathbb{Z}_{\ge 0}$, and finite constant $C>0$, such that:
\begin{align*}
|\mathscr{T}(f)| \leq C \sum_{|\alpha| \leq N} \ \sup_{p \in\mathbb{R}^{1,3}} \left( 1+ \|p\| \right)^{N}|D^{\alpha}f(p)|, \hspace{5mm} \forall f\in \mathcal{S}(\mathbb{R}^{1,3})
\end{align*} 
where $\|p\|^{2} = \sum_{\mu=0}^{3} |p_{\mu}|^{2}$, $D^{\alpha} = \frac{\partial^{|\alpha|}}{ (\partial p_{0})^{\alpha_{0}} \cdots  (\partial p_{3})^{\alpha_{3}}}$, and: $|\alpha| = \alpha_{0} + \alpha_{1} + \alpha_{2} + \alpha_{3}$
\label{bros_epstein_glaser}
\end{theorem}
\ \\

\noindent
Theorem~\ref{bros_epstein_glaser} is in fact a special case of the boundedness condition satisfied by all tempered distributions~\cite{Strichartz94}. $N$ corresponds to the highest number of derivatives $|\alpha|$ which appear in this bound. It is important to note though that $N$ is not necessarily unique -- a distribution may have different representations which each satisfy boundedness conditions with different numbers of derivatives. A more meaningful parameter is the \textit{minimal} value of $N$ over all possible representations, which is called the \textit{order} of $\mathscr{T}$~\cite{Strichartz94}. In Theorem~\ref{cdt} one is interested in the \textit{leading} asymptotic distance behaviour, and so $N$ in this theorem corresponds to the order of $\widehat{\mathcal{T}}_{(1,2)}^{T}(p)$. To be consistent with Theorem~\ref{cdt}, throughout the rest of this paper it will be implicitly assumed that $N$ in Theorem~\ref{bros_epstein_glaser} corresponds to the order of $\mathscr{T}$. \\ 

\noindent
In light of Eq.~(\ref{F_corr}), it is clearly important to establish whether one can determine a condition for when $N=0$. By applying Theorem~\ref{bros_epstein_glaser}, one in fact has the following necessary and sufficient condition:
\ \\
\begin{theorem}
Given that $\mathscr{T} \in \mathcal{S}'(\mathbb{R}^{1,3})$ has support in $\overbar{V}^{+}$
\begin{align*}
N=0 \hspace{3mm} \Longleftrightarrow \hspace{3mm} \textit{$\mathscr{T}$ defines a finite measure} 
\end{align*} 
\label{N_measure}
\end{theorem}
\begin{proof}[Proof {\bf{$(\Longrightarrow)$:}}] 
Since $N=0$, from Theorem~\ref{bros_epstein_glaser}:
\begin{align*}
|\mathscr{T}(f)| \leq C \sup_{p \in\mathbb{R}^{1,3}} |f(p)|, \hspace{5mm} \forall f\in \mathcal{S}(\mathbb{R}^{1,3})
\end{align*}  
for some finite constant $C>0$. The linear functional $\mathscr{T}$ is therefore bounded, and hence continuous with respect to the topology induced by the norm $\|\cdot\|_{\infty}:= \sup_{p \in\mathbb{R}^{1,3}} |\cdot|$. The restricted map $\mathscr{T}|_{C_{c}^{\infty}(\mathbb{R}^{1,3})}$ satisfies the same bound as $\mathscr{T}$ $\forall f\in C_{c}^{\infty}(\mathbb{R}^{1,3}) = \mathcal{S}(\mathbb{R}^{1,3})\cap C_{c}^{0}(\mathbb{R}^{1,3})$, and since $C_{c}^{\infty}(\mathbb{R}^{1,3})$ with the norm $\|\cdot\|_{\infty}$ is a linear subspace of the normed vector space $\left(C_{c}^{0}(\mathbb{R}^{1,3}),\|\cdot\|_{\infty}\right)$, $\mathscr{T}|_{C_{c}^{\infty}(\mathbb{R}^{1,3})}$ is also linear. Conversely, if a linear map $\mathscr{T}|_{C_{c}^{\infty}(\mathbb{R}^{1,3})}$ satisfies a tempered distribution bound, then this same bound must \textit{also} hold $\forall f\in \mathcal{S}(\mathbb{R}^{1,3})$, and so in this sense $\mathscr{T}|_{C_{c}^{\infty}(\mathbb{R}^{1,3})}$ (uniquely) defines $\mathscr{T}$~\cite{Strichartz94}. This means that one can think of $\mathscr{T}|_{C_{c}^{\infty}(\mathbb{R}^{1,3})}$ and $\mathscr{T}$ as corresponding to the \textit{same} tempered distribution. As a consequence of the \textit{Hahn-Banach Theorem}~\cite{Zhelobenko06} it follows that $\mathscr{T}|_{C_{c}^{\infty}(\mathbb{R}^{1,3})}$ (and hence $\mathscr{T}$) can be extended to a linear functional $\widetilde{\mathscr{T}}: C_{c}^{0}(\mathbb{R}^{1,3}) \rightarrow \mathbb{C}$ which satisfies the same bound: 
\begin{align*}
|\widetilde{\mathscr{T}}(g)| \leq C \|g\|_{\infty}, \hspace{5mm} \forall g\in C_{c}^{0}(\mathbb{R}^{1,3})
\end{align*}  
Moreover, due to this bound and the linearity of $\widetilde{\mathscr{T}}$ it follows that this extension is unique~\cite{Al-Gwaiz92}. Since $\widetilde{\mathscr{T}}$ is a bounded linear functional on $C_{c}^{0}(\mathbb{R}^{1,3})$, the \textit{Riesz Representation Theorem}~\cite{Rudin70} implies that there exists a unique finite (complex) regular Borel measure $\mu$ on $\mathbb{R}^{1,3}$ defined by:
\begin{align*}
\widetilde{\mathscr{T}}(g) = \int_{\mathbb{R}^{1,3}}  g \ d\mu
\end{align*}
In this sense $\widetilde{\mathscr{T}}$ is said to define a finite measure $\mu$. Since $\widetilde{\mathscr{T}}$ is an unique extension which is entirely specified by $\mathscr{T}$, one defines $\widetilde{\mathscr{T}} \equiv \mathscr{T}$ with the understanding that when $\mathscr{T}$ acts on functions in $C_{c}^{0}(\mathbb{R}^{1,3})\backslash C_{c}^{\infty}(\mathbb{R}^{1,3})$, this is defined by the specific (unique) limit of $\mathscr{T}$ acting on a sequence of regularised test functions~\cite{Al-Gwaiz92}.
\end{proof}
\begin{proof}[Proof {\bf{$(\Longleftarrow)$:}}] 
Assuming $\mu$ is a finite measure implies~\cite{Baake_Grimm13} that $\mathscr{T}$ satisfies the bound:
\begin{align*}
|\mathscr{T}(g)| \leq \widetilde{C}\|g\|_{\infty}, \hspace{5mm} \forall g\in C_{c}^{0}(\mathbb{R}^{1,3})
\end{align*} 
for some finite $\widetilde{C}>0$, and therefore the restricted map $\mathscr{T}|_{C_{c}^{\infty}(\mathbb{R}^{1,3})}$ also obeys this bound. From the discussion in the proof in the ($\Longrightarrow$) direction, $\mathscr{T}|_{C_{c}^{\infty}(\mathbb{R}^{1,3})}$ satisfying this bound is sufficient to imply that $\mathscr{T}$ must also have the same bound, but instead $\forall g \in \mathcal{S}(\mathbb{R}^{1,3})$. Since $|\mathscr{T}(g)| \leq \widetilde{C}\|g\|_{\infty} \, \forall g\in \mathcal{S}(\mathbb{R}^{1,3})$, it then follows from Theorem~\ref{bros_epstein_glaser} that $N=0$.    
\end{proof}

\ \\
\noindent
In the context of the cluster decomposition theorem (Theorem~\ref{cdt}), Theorem~\ref{N_measure} implies that $N=0$ if and only if $\widehat{\mathcal{T}}_{(1,2)}^{T}(p)$ defines a finite measure. So establishing whether $\widehat{\mathcal{T}}_{(1,2)}^{T}(p)$ defines a finite measure or not is the key to determining how the correlation strength $F^{\phi_{1}\phi_{2}}(r)$ between the clusters of fields $\phi_{1}$ and $\phi_{2}$ changes as $r \rightarrow \infty$. Since $\widehat{\mathcal{T}}_{(1,2)}(p)$ satisfies Eq.~(\ref{smear_p}), one has the following proposition:
\ \\ 
\begin{prop}
If $\widehat{\mathcal{T}}_{(1,2)}^{T}$ defines a measure then this measure is finite.
\label{paper4_prop1}
\end{prop} 
\begin{proof}
The condition for a tempered distribution to be finite as a measure is that the limit $\lim_{k \rightarrow \infty}|\widehat{\mathcal{T}}_{(1,2)}^{T}(\psi_{k})|$ is finite, where $\psi_{k} \in \mathcal{D}(\mathbb{R}^{1,3})$ are a sequence of test functions such that $\lim_{k \rightarrow \infty} \psi_{k} = 1$. Computing this limit one has:
\begin{align*}
\lim_{k \rightarrow \infty}|\widehat{\mathcal{T}}_{(1,2)}^{T}(\psi_{k})| = \lim_{k \rightarrow \infty}|g\widehat{T}_{(1,2)}^{T}(\psi_{k})| = \lim_{k \rightarrow \infty}|\widehat{T}_{(1,2)}^{T}(g\psi_{k})| = |\widehat{T}_{(1,2)}^{T}(g)| < \infty
\end{align*}  
where the final two identities follow from the continuity and boundedness of $\widehat{T}_{(1,2)}^{T}$ respectively.
\end{proof}

\noindent
So in contrast to $\widehat{T}_{(1,2)}^{T}$, which can potentially define either a finite or an unbounded measure, measures defined by $\widehat{\mathcal{T}}_{(1,2)}^{T}$ can only be finite, independent of the form of the test functions $f^{(i)}$ used to define the clusters. Now if one combines Theorem~\ref{N_measure} with Proposition~\ref{paper4_prop1}, this implies the important corollary:
\ \\
\begin{corr}
Given that $\widehat{\mathcal{T}}_{(1,2)}^{T} \in \mathcal{S}'(\mathbb{R}^{1,3})$ has support in $\overbar{V}^{+}$, and is the Fourier transform of a truncated cluster correlator
\begin{align*}
N=0 \hspace{3mm} \Longleftrightarrow \hspace{3mm} \textit{$\widehat{\mathcal{T}}_{(1,2)}^{T}$ defines a measure} 
\end{align*}
\label{N_corr}
\end{corr}  

\noindent
Thus in order to determine if the Fourier transform of a such a cluster correlator has $N=0$, one is only required to prove that the tempered distribution $\widehat{\mathcal{T}}_{(1,2)}^{T}$ defines a measure, since its finiteness is guaranteed. Conversely, if $\widehat{\mathcal{T}}_{(1,2)}^{T}$ can be shown to not define a measure, this is sufficient to prove that $N \neq 0$. A similar theorem was previously established by~\cite{Strocchi78}, in which the condition $N \neq 0$ is linked to the failure of Fourier transformed cluster correlators to define measures in some suitable neighbourhood of the light cone $\{p^{2}=0\}$. Corollary~\ref{N_corr} is a generalisation of this theorem, since no restrictions are imposed on the neighbourhood in which $\widehat{\mathcal{T}}_{(1,2)}^{T}$ should fail to be a measure, only that it should fail to be a measure on $\mathbb{R}^{1,3}$. Moreover, the theorem by~\cite{Strocchi78} implicitly assumes that $\mathcal{V}$ has no mass gap, whereas the proof of Theorem~\ref{N_measure} (and hence Corollary~\ref{N_corr}) only depends on the support property of $\widehat{T}_{(1,2)}^{T}$. \\

\noindent
When the structural form of correlators in QFTs are analysed and discussed, this is done almost exclusively with respect to the correlators involving the basic fields, and hence $\widehat{T}_{(1,2)}^{T}$. However, Corollary~\ref{N_corr} emphasises that one can only determine if the CDP is violated or not if the structure of $\widehat{\mathcal{T}}_{(1,2)}^{T}$ is known. It is therefore important to establish what conditions $\widehat{\mathcal{T}}_{(1,2)}^{T}$ must satisfy in order to ensure that $\widehat{T}_{(1,2)}^{T}$ defines a measure, and vice versa. In this regard, one has the proposition:
\ \\
\begin{prop}
If \ $\widehat{T}_{(1,2)}^{T}$ defines a measure \hspace{2mm} $\Longrightarrow$ \hspace{2mm} $\widehat{\mathcal{T}}_{(1,2)}^{T}$ defines a measure
\label{paper4_prop2}
\end{prop}
\begin{proof}
$\widehat{\mathcal{T}}_{(1,2)}^{T}$ is a (linear) functional on $C_{c}^{0}(\mathbb{R}^{1,3})$ because: $\widehat{\mathcal{T}}_{(1,2)}^{T}(f) = g\widehat{T}_{(1,2)}^{T}(f) = \widehat{T}_{(1,2)}^{T}(gf)$ and $\widehat{T}_{(1,2)}^{T}$ is always well-defined on $gf$ since $gf \in C_{c}^{0}(\mathbb{R}^{1,3})$ for $g \in \mathcal{S}(\mathbb{R}^{1,3})$ and $f \in C_{c}^{0}(\mathbb{R}^{1,3})$. Moreover, since $g$ is bounded and $\widehat{T}_{(1,2)}^{T}$ defines a measure, it follows that for every compact set $K \subset \mathbb{R}^{1,3}$: 
\begin{align*}
|\widehat{\mathcal{T}}_{(1,2)}^{T}(f)| = |\widehat{T}_{(1,2)}^{T}(gf)| \leq C_{K} \sup_{p \in K} |g(p)f(p)| \leq C_{K} \sup_{p \in K}|g(p)| \cdot \sup_{p \in K}|f(p)| \leq \widetilde{C}_{K}\sup_{p \in K}|f(p)|
\end{align*}  
holds for all $f$ with support in $K$, and so $\widehat{\mathcal{T}}_{(1,2)}^{T}$ defines a measure.   
\end{proof}

\noindent
It is also clearly important to understand to what extent Proposition~\ref{paper4_prop2} holds in the opposite direction, or equivalently: what conditions must $\widehat{T}_{(1,2)}^{T}$ satisfy in order to ensure that $\widehat{\mathcal{T}}_{(1,2)}^{T}$ does not define a measure? One example of such a condition is summarised by the following proposition:
\ \\
\begin{prop}
Let $\sigma$ be a (tempered) distribution with discrete support, and $D \sigma$ the distributional derivative of $\sigma$. 
\begin{align*}
\text{$\widehat{T}_{(1,2)}^{T} = D \sigma$ \hspace{2mm} $\Longrightarrow$ \hspace{2mm} $\widehat{\mathcal{T}}_{(1,2)}^{T}$ does not define a measure}
\end{align*}
\label{paper4_prop3}
\end{prop}
\begin{proof}
In general, if $\sigma$ has discrete support $\mathbb{S} \subset \mathbb{R}^{1,3}$, then $\sigma = \sum_{p \in\mathbb{S}}\sum_{\alpha}a_{\alpha} D^{\alpha} \delta_{p}$, where both sums are finite~\cite{Strichartz94}. Assuming $f \in C_{c}^{0}(\mathbb{R}^{1,3})$, one has: 
\begin{align*}
\widehat{\mathcal{T}}_{(1,2)}^{T}(f) = g\widehat{T}_{(1,2)}^{T}(f) = \widehat{T}_{(1,2)}^{T}(gf) = D \sigma (gf) &= -\sigma (D(gf)) \\
&= -\sum_{p \in\mathbb{S}}\sum_{\alpha}a_{\alpha} D^{\alpha} \delta_{p}(D(gf)) \\
&= (-1)^{\alpha +1}\sum_{p \in\mathbb{S}}\sum_{\alpha}a_{\alpha} \delta_{p}\left[D^{\alpha}(fD(g) + gD(f))\right] 
\end{align*}
So the value of $\widehat{\mathcal{T}}_{(1,2)}^{T}(f)$ always depends on the derivative $D(f)$ evaluated at the points $p \in\mathbb{S}$. However, if $f$ is not differentiable at some $p \in\mathbb{S}$, then $D(f)(p)$ is ill-defined, which proves that $\widehat{\mathcal{T}}_{(1,2)}^{T}$ cannot be a functional on $C_{c}^{0}(\mathbb{R}^{1,3})$, and therefore does not define a measure.         
\end{proof}
\ \\

\noindent
Together, Propositions~\ref{paper4_prop2} and~\ref{paper4_prop3} describe how the properties of $\widehat{T}_{(1,2)}^{T}(p)$ can be used to establish whether $\widehat{\mathcal{T}}_{(1,2)}^{T}(p)$ defines a measure, and thus by Corollary~\ref{N_corr} whether $N$ is vanishing or not. Since $\widehat{T}_{(1,2)}^{T}(p)$ is constructed purely in terms of the basic fields $\phi_{k}$, the structure of $\widehat{T}_{(1,2)}^{T}(p)$ is constrained by Lorentz symmetry, which implies that $\widehat{T}_{(1,2)}^{T}(p)$ can be written in a general form. This form will be outlined in the next section, and its connection to the measure properties of $\widehat{T}_{(1,2)}^{T}(p)$ will be described.

\section{The spectral structure of QFT correlators}
\label{paper4_section2}

\subsection{The spectral representation}
\label{paper4_section2_1}

In axiomatic formulations of QFT, both $\widehat{T}_{(1,2)}^{T}(p)$ and $\widehat{T}_{(1,2)}(p)$ are \textit{Lorentz covariant} tempered distributions, and therefore satisfy the following condition~\cite{Bogolubov_Logunov_Oksak90}: 
\begin{align}
\widehat{T}(\Lambda p) = S(\Lambda) \, \widehat{T}(p), \hspace{10mm} \Lambda \in \overbar{\mathscr{L}_{+}^{\uparrow}} \cong \mathrm{SL}(2,\mathbb{C})
\end{align}     
where $\overbar{\mathscr{L}_{+}^{\uparrow}}$ is the universal cover of the identity component of the Lorentz group\footnote{This group consists of Lorentz transformations that preserve both orientation and the direction of time.} $\mathscr{L}_{+}^{\uparrow}$, and $S$ is some (finite-dimensional) representation of $\mathrm{SL}(2,\mathbb{C})$. In the special case where the representation is trivial ($S(\Lambda)\equiv 1$), $\widehat{T}(p)$ is called a \textit{Lorentz invariant} distribution. The structure of the Lorentz covariant distribution $\widehat{T}(p)$ is dependent upon how the fields $\phi_{1}$ and $\phi_{2}$ transform under Lorentz transformations. In particular, $\widehat{T}(p)$ has the following decomposition~\cite{Bogolubov_Logunov_Oksak90}:
\begin{align}
\widehat{T}(p) = \sum_{\alpha=1}^{\mathscr{N}}Q_{\alpha}(p) \, \widehat{T}_{\alpha}(p)
\label{cov_decomp}
\end{align}    
where $\widehat{T}_{\alpha}(p)$ are Lorentz invariant distributions, and $Q_{\alpha}(p)$ are Lorentz covariant polynomial functions of $p$ which carry the Lorentz index structure of $\phi_{1}$ and $\phi_{2}$. The simplest case is the Fourier transform of a correlator involving two scalar fields. Since both the fields are scalar it follows that $Q_{1}(p) = 1$, and therefore $\widehat{T}(p) = \widehat{T}_{1}(p)$. As an example, consider $\widehat{D}(p) = \mathcal{F}\left[\langle 0| \phi(x)\phi(y)|0\rangle \right]$, where $\phi$ is a free scalar field of mass $m$. In this case $\widehat{D}(p)$ has the explicit form:
\begin{align}
\widehat{D}(p) =  a\, \delta(p) + 2\pi \theta(p^{0})\delta(p^{2} - m^{2}) 
\label{KL_rep_scalar} 
\end{align}  
where $a = |\langle 0|\phi |0\rangle|^{2}$, and hence: $\widehat{D}^{T}(p)= 2\pi \theta(p^{0})\delta(p^{2} - m^{2})$. The truncation of the scalar correlator therefore removes the component of $\widehat{D}(p)$ concentrated at $p=0$. As expected from Eq.~(\ref{cov_decomp}), the overall structure of $\widehat{D}(p) = \widehat{D}_{1}(p)$ (and also $\widehat{D}^{T}(p)$) is Lorentz invariant. \\

\noindent
Now consider the case where $\widehat{S}(p)$ is the Fourier transform of a correlator that involves a Dirac spinor and conjugate spinor field. Here there exist two possible Lorentz covariant polynomials: $Q_{1}(p) = \mathbb{I}$ and $Q_{2}(p) = \gamma^{\mu}p_{\mu} = \slashed{p}$, and hence:
\begin{align}
\widehat{S}(p) = \mathbb{I} \, \widehat{S}_{1}(p) + \slashed{p} \, \widehat{S}_{2}(p)
\label{ferm_decomp}
\end{align}     
where the spinor indices have been suppressed. In the case where $\widehat{S}(p) = \mathcal{F}\left[\langle 0| \psi(x)\overline{\psi}(y)|0\rangle \right]$, and $\psi$ is a free Dirac field of mass $m$, $\widehat{S}(p)$ has the explicit form:
\begin{align}
\widehat{S}(p) &=  2\pi (\slashed{p} + m) \theta(p^{0})\delta(p^{2} - m^{2}) 
\label{KL_rep_ferm} 
\end{align} 
Comparing this expression with Eq.~(\ref{ferm_decomp}) it follows that $\widehat{S}_{1}(p) = 2\pi m \, \theta(p^{0})\delta(p^{2} - m^{2})$, and $\widehat{S}_{2}(p) = 2\pi \theta(p^{0})\delta(p^{2} - m^{2})$, which again as anticipated are both Lorentz invariant. Finally, consider a Fourier transformed correlator $\widehat{D}_{\mu\nu}(p)$ that involves two arbitrary vector fields with indices $\mu$ and $\nu$ respectively. In this case there are also two possible Lorentz covariant polynomials: $Q_{1}(p) = g_{\mu\nu}$ and $Q_{2}(p) = p_{\mu}p_{\nu}$, which implies:
\begin{align}
\widehat{D}_{\mu\nu}(p) = g_{\mu\nu} \, \widehat{D}_{1}(p) + p_{\mu}p_{\nu} \, \widehat{D}_{2}(p)
\label{vector_decomp}
\end{align}  
An example of this class of Lorentz covariant distributions is: $\widehat{D}_{\mu\nu}(p) = \mathcal{F}\left[\langle 0| A_{\mu}(x)A_{\nu}(y)|0\rangle \right]$, where $A_{\mu}$ is a free (abelian) gauge field. One can demonstrate that $\widehat{D}_{\mu\nu}(p)$ has the form~\cite{Bogolubov_Logunov_Oksak90}: 
\begin{align}
\widehat{D}_{\mu\nu}(p) = 2\pi \theta(p^{0})\left[ -g_{\mu\nu}\delta(p^{2}) + (\xi-1)\delta'(p^{2})p_{\mu}p_{\nu} \right]
\label{KL_rep_vec}
\end{align} 
where $\xi$ is the gauge-fixing parameter which is introduced in order to locally quantise the theory. Comparing this expression with Eq.~(\ref{vector_decomp}), $\widehat{D}_{1}(p) = -2\pi \theta(p^{0})\delta(p^{2})$ and $\widehat{D}_{2}(p) =  2\pi \theta(p^{0})(\xi-1)\delta'(p^{2})$, which are both Lorentz invariant. \\   

\noindent
Now that the general structure of Lorentz covariant (Fourier transformed) correlators in QFTs has been outlined, the results of Sec.~\ref{paper4_section1} can be applied to some explicit physical examples. Firstly, consider the free quantised electromagnetic field $A_{\mu}$. The Fourier transform of the correlator of this field is given by Eq.~(\ref{KL_rep_vec}). The non-gauge-fixing component of $\widehat{D}_{\mu\nu}(p)$ clearly defines a measure, whereas for general values of $\xi$, the gauge-fixing component contains a $\delta'(p^{2})$ term, which does not define a measure. However, since the physical characteristics of gauge theories are independent of the value of the $\xi$, one is free to set $\xi=1$, and then $\widehat{D}_{\mu\nu}(p) = -2\pi g_{\mu\nu}\theta(p^{0})\delta(p^{2})$. Since\footnote{The equality $\langle 0 | A_{\mu}(x)A_{\nu}(y)|0\rangle^{T} = \langle 0 | A_{\mu}(x)A_{\nu}(y)|0\rangle$ follows from Lorentz invariance~\cite{Strocchi13}.} $\langle 0 | A_{\mu}(x)A_{\nu}(y)|0\rangle^{T} = \langle 0 | A_{\mu}(x)A_{\nu}(y)|0\rangle$, and $\widehat{D}_{\mu\nu}^{T}(p)=\widehat{D}_{\mu\nu}(p)$ defines a measure, it follows from Proposition~\ref{paper4_prop2} and Corollary~\ref{N_corr} that the CDP is preserved ($N=0$). Due to Theorem~\ref{cdt}, physically this implies that the correlation strength $F_{\text{free}}^{\gamma\gamma}(r)$ between clusters containing free photons has the asymptotic behaviour:
\begin{align}
F_{\text{free}}^{\gamma\gamma}(r)\sim \frac{1}{r^{2}}, \hspace{5mm} r \rightarrow \infty 
\label{photon_corr} 
\end{align}
and hence the correlation between free photons becomes increasingly weaker the further away they are separated. \\

\noindent
In QCD, the correlators of interest involve quark and gluon fields. Unlike the photon field, the non-abelian gluon field $A_{\mu}^{a}$ can never be free because there are always (self) interactions, even with the absence of quark fields. It is possible though to consider free quarks, and in fact the Fourier transformed free quark correlator $\widehat{S}(p)$ has the form of Eq.~(\ref{KL_rep_ferm}), where now $m$ is the mass of the quark. Because\footnote{$\langle 0| \psi(x)\overline{\psi}(y)|0\rangle^{T}=\langle 0| \psi(x)\overline{\psi}(y)|0\rangle$ also because of Lorentz invariance.} $\langle 0| \psi(x)\overline{\psi}(y)|0\rangle^{T}=\langle 0| \psi(x)\overline{\psi}(y)|0\rangle$, and $\widehat{S}^{T}(p)=\widehat{S}(p)$ defines a measure, Proposition~\ref{paper4_prop2} and Corollary~\ref{N_corr} imply that $N=0$. Since the theory contains a mass gap $(0,m)$, it follows from Theorem~\ref{cdt} that the correlation strength $F_{\text{free}}^{q\bar{q}}(r)$ between free quarks behaves asymptotically as:
\begin{align}
F_{\text{free}}^{q\bar{q}}(r)\sim \frac{e^{-mr}}{r^{\frac{3}{2}}}, \hspace{5mm} r \rightarrow \infty  
\end{align} 
and hence free quarks preserve the CDP. Physically, this means that the correlation between the quarks is exponentially suppressed at large distances, which supports the hypothesis that free quarks are not confined, as one would expect. \\

\noindent
Although the examples considered so far have consisted of exactly solvable free theories, the results derived in Sec.~\ref{paper4_section1} are also equally applicable to interacting QFTs. However, in general the analytic structure of correlators is poorly understood for interacting theories, and so it is more difficult to prove directly whether the CDP is preserved or not. Quantum electrodynamics (QED) is one example of an interacting theory though, where it is possible to determine this property~\cite{Nakanishi_Ojima90}. In QED, $F_{\mu\nu}$ is an observable, which implies that $\langle 0| F_{\mu\nu}(x)F_{\rho\sigma}(y) |0\rangle$ is a positive-definite distribution\footnote{A positive-definite distribution $T(x-y)$ satisfies the condition: $\int d^{4}xd^{4}y \, T(x-y) \overbar{f(x)}f(y) \geq 0$ for any test function $f$.}~\cite{Nakanishi_Ojima90}. Positive-definiteness is sufficient to prove that $\mathcal{F}\left[\langle 0| A_{\mu}(x)A_{\nu}(y) |0\rangle^{T} \right]$ is non-negative, defines a measure~\cite{Bogolubov_Logunov_Oksak90}, and thus by Proposition~\ref{paper4_prop2} $N=0$, which implies that interacting photons also preserve the CDP. In QCD though, this same argument fails because $F_{\mu\nu}^{a}$ is \textit{not} an observable, which itself is a consequence of the non-abelian nature of the theory. The failure of this proof is certainly suggestive that the CDP may no longer hold for the gluon cluster correlator in QCD. However, as emphasised by~\cite{Roberts_Williams_Krein91}, the failure of this proof is not same as a proof of its failure, and it remains an open question as to whether the gluon correlator, or more generally cluster correlators of coloured fields, violate the CDP. The problem in QCD is that the precise form of correlators involving coloured fields is controlled by the non-perturbative dynamics of the theory, which is \textit{a priori} unknown. Nevertheless, the general non-perturbative structure of these correlators can still be characterised in a general manner, and this will be discussed in the following section.

\subsection{The spectral density}
\label{paper4_section2_2}

Due to Eq.~(\ref{cov_decomp}), it is clear that in order to characterise Fourier transformed correlators, it is important to understand the general structure of Lorentz invariant distributions $\widehat{T}_{\alpha}$. As discussed in Sec.~\ref{paper4_section1}, these distributions are also required to have support in the closed forward light cone $\overbar{V}^{+}$. If a tempered distribution $\widehat{T}_{\alpha} \in \mathcal{S}'(\mathbb{R}^{1,3})$ is both Lorentz invariant, and has support in $\overbar{V}^{+}$, it turns out that $\widehat{T}_{\alpha}$ can be written in the following general manner~\cite{Bogolubov_Logunov_Oksak90}:   
\begin{align}
\widehat{T}_{\alpha}(p) =  P(\partial^{2})\delta(p) + \int_{0}^{\infty} ds \, \theta(p^{0})\delta(p^{2}-s) \rho_{\alpha}(s) 
\label{KL_gen_rep} 
\end{align}   
where $P(\partial^{2})$ is some polynomial in the d'Alembert operator $\partial^{2} = g_{\mu\nu}\frac{\partial}{\partial p_{\mu}}\frac{\partial}{\partial p_{\nu}}$, and $\rho_{\alpha}(s) \in \mathcal{S}'(\bar{\mathbb{R}}_{+})$. This important structural relation is called the \textit{spectral representation} of $\widehat{T}_{\alpha}$, and $\rho_{\alpha}$ is referred to as the \textit{spectral density}. If $\widehat{T}_{\alpha}$ is also a non-negative distribution\footnote{A non-negative distribution $\widehat{T}_{\alpha}$ has the property that $\widehat{T}_{\alpha}(f) \geq 0$, $\forall f \geq 0$.}, it follows that $\widehat{T}_{\alpha}(f) = \int f(p) \, d\mu(p)$, where $d\mu(p)$ is a non-negative (tempered) measure. In this case $\widehat{T}_{\alpha}(p)$ has the form:
\begin{align}
\widehat{T}_{\alpha}(p) =  c\, \delta(p) + \int_{0}^{\infty} ds \, \theta(p^{0})\delta(p^{2}-s) \rho_{\alpha}(s) 
\label{KL_rep} 
\end{align} 
where $c \geq 0$, and $\rho_{\alpha}(s)$ defines a (tempered) measure $d\rho_{\alpha}(s) = \rho_{\alpha}(s)ds$. Eq.~(\ref{KL_rep}) is called the \textit{K\"all\'{e}n-Lehmann representation}. In fact, independently of whether $\widehat{T}_{\alpha}$ is non-negative or not, if $\widehat{T}_{\alpha}$ defines a measure then this is sufficient to imply that $\widehat{T}_{\alpha}$ must have the form of Eq.~(\ref{KL_rep}), but $c$ is not necessarily positive in this case. \\

\noindent
The spectral representation of $\widehat{T}_{\alpha}$ emphasises that the $\mathscr{N}$ spectral densities $\rho_{\alpha}$ associated with a correlator play a fundamental role in determining its structure. In fact, when $\widehat{T}_{\alpha}$ is a measure, and Eq.~(\ref{KL_rep}) holds, correlators are uniquely specified by $\{ \rho_{\alpha} \}$. Comparing the Fourier transformed free correlator components $\widehat{S}_{i}(p)$ and $\widehat{D}_{i}(p)$ discussed in Sec.~\ref{paper4_section2_1} with Eq.~(\ref{KL_gen_rep}), one can directly see that the spectral densities for the Dirac correlator are:  
\begin{align}
\rho_{1}^{\psi}(s)= 2\pi m \, \delta(s-m^{2}), \hspace{5mm} \rho_{2}^{\psi}(s)= 2\pi \delta(s-m^{2}) 
\label{fermion_spec} 
\end{align}
and for the vector field correlator are:
\begin{align}
\rho_{1}^{A}(s) = -2\pi \delta(s), \hspace{5mm} \rho_{2}^{A}(s) = 2\pi (\xi-1) \delta'(s)  
\label{vector_spec}
\end{align}
In Eq.~(\ref{fermion_spec}), both of the spectral densities define measures. This is not surprising because $\widehat{S}(p)$ itself defines a measure, and therefore both $\widehat{S}_{1}(p)$ and $\widehat{S}_{2}(p)$ must satisfy Eq.~(\ref{KL_rep}). In Eq.~(\ref{vector_spec}) though, $\rho_{1}^{A}$ defines a measure but $\rho_{2}^{A}$ clearly does not, unless $\xi=1$. That $\rho_{2}^{A}$ is permitted to not define a measure is a symptom of the fact that much like QCD, a free quantised abelian gauge theory is a locally quantised QFT. As discussed in Sec.~\ref{paper4_intro}, this means that the space of states $\mathcal{V}$ no longer has a positive-definite inner product. This implies, among other things, that the Fourier transformed correlators are not guaranteed to be non-negative, and so may not define measures and satisfy Eq.~(\ref{KL_rep})~\cite{Bogolubov_Logunov_Oksak90}. \\

\noindent
It is clear from the results in this section that the behaviour of correlators is closely connected to the structure of the associated spectral densities $\{ \rho_{\alpha} \}$. In light of Eqs.~(\ref{cov_decomp}) and~(\ref{KL_rep}) it is clear that if any of the spectral densities $\rho_{\alpha}$ does not define a measure, then this is sufficient to imply that $\widehat{T}$ must also not define a measure. However, as already mentioned in Sec.~\ref{paper4_section1}, the failure of $\widehat{T}_{(1,2)}^{T}$ to define a measure may not imply that $\widehat{\mathcal{T}}_{(1,2)}^{T}$ does not define a measure, and that $N \neq 0$. But by Proposition~\ref{paper4_prop3}, if $\widehat{T}_{(1,2)}^{T}=D\sigma$ where $\sigma$ is a distribution with discrete support, then this \textit{is} sufficient to imply that $N \neq 0$. In fact, if any of the components in the Lorentz covariant decomposition of $\widehat{T}_{(1,2)}^{T}$ contains a term of the form $D\sigma$, then this implies that $N \neq 0$. The simplest such example is if $\widehat{T}_{(1,2)}^{T}(p)= \theta(p^{0})\delta'(p^{2}-a)$, and hence by Eq.~(\ref{KL_gen_rep}), $\rho(s) = \delta_{a}'(s) =\delta'(s-a)$, as is the case for $\rho_{2}^{A}$ in Eq.~(\ref{vector_spec}) (with $\xi \neq 1$ and $a=0$). Therefore, if the spectral densities of any (truncated) correlator contain a $\delta_{a}'$ term, this is sufficient to imply that $N\neq 0$. This result is particularly interesting when applied to QCD because it provides a definite condition with which to test whether truncated cluster correlators of coloured fields violate the CDP. Both this condition and its applications will be discussed in the next section.

\section{The cluster decomposition property in QCD}
\label{paper4_section3}
 
As outlined at the end of the Sec.~\ref{paper4_section2_2}, determining the structure of the spectral densities of correlators that involve coloured fields in QCD, such as the quark and gluon correlators, is a direct way to establish whether $N=0$ or not. In particular, if any of the spectral densities contain a $\delta_{a}'$ term, this is sufficient to imply that $N \neq 0$. Combining this result with Theorem~\ref{cdt}, one has the following corollary:
\ \\
\begin{corr}
Assuming that the state space $\mathcal{V}_{\text{QCD}}$ has no mass gap, and that any of the spectral densities $\{\rho_{\alpha}\}$ of correlators involving coloured fields contains a $\delta_{a}'$ term, this is sufficient to ensure confinement. 
\label{qcd_spectral}  
\end{corr}
\ \\
\noindent
It is important to emphasise here that the requirement of no mass gap refers to the full state space $\mathcal{V}$ (with an indefinite inner product), and not to the physical Hilbert space $\mathcal{H}$. In particular, this means that it is possible for $\mathcal{V}$ to have no mass gap but $\mathcal{H}$ to have one, which is expected to be the case in QCD~\cite{Nakanishi_Ojima90}. Since the precise analytic form of spectral densities (and hence correlators) are unknown in QCD, it is difficult to establish whether $\{\rho_{\alpha} \}$ contain non-measure defining terms or not. Nevertheless, by using non-perturbative methods such as lattice QFT it is possible to calculate approximations to these objects using numerical fits. An important expression in this regard is the so-called \textit{Schwinger function} $C_{\alpha}(t)$, which can be written in the following form\footnote{In general, the Schwinger function is defined by: $C_{\alpha}(t) = \frac{1}{2\pi}\int_{-\infty}^{\infty} dp_{0} \, e^{ip_{0}t}\Delta_{\alpha}(p^{2})|_{\mathbf{p}=0}$, where $\Delta_{\alpha}(p^{2})$ is one of the components of a Euclidean propagator. It should be noted that $C_{\alpha}(t)$ reduces to the form of Eq.~(\ref{schwinger}) only if one assumes that $\Delta_{\alpha}(p^{2})$ does not contain $P(\partial^{2})\delta(p)$ contributions as in Eq.~(\ref{KL_gen_rep}).}:
\begin{align}
C_{\alpha}(t) = \int_{0}^{\infty} ds \, \rho_{\alpha}(s) \frac{e^{-\sqrt{s}t}}{2\sqrt{s}}
\label{schwinger}
\end{align}         
where $\rho_{\alpha}$ is the spectral density associated with one of the components of a specific propagator\footnote{The spectral densities $\{\rho_{\alpha} \}$ which define correlators are the same as those that define propagators. Moreover, $\{\rho_{\alpha} \}$ have the same form for both Euclidean and Minkowski spacetime propagators.}, and $t \geq 0$. By calculating $C_{\alpha}(t)$ on the lattice, this provides a way of indirectly probing the structure of $\rho_{\alpha}$. $C_{\alpha}(t)$ has been computed for both the quark and gluon propagators, and one of the most striking features is that $C_{\alpha}(t)$ appears to become negative at some value of $t$~\cite{Alkofer_Detmold_Fischer_Maris04,Silva_Oliveira_Dudal_Bicudo_Cardoso13}. This feature implies that $\rho_{\alpha}$ violates non-negativity, which is usually interpreted as evidence of confinement~\cite{Alkofer_Detmold_Fischer_Maris04,Silva_Oliveira_Dudal_Bicudo_Cardoso13}. However, it is often assumed that $\rho_{\alpha}(s)$ makes sense as a function, and that the negativity of $C_{\alpha}(t)$ is due to $\rho_{\alpha}(s)$ becoming negative over some continuous range of $s$~\cite{Cornwall13}. In general though, because $\rho_{\alpha}$ is a distribution (in $\mathcal{S}'(\bar{\mathbb{R}}_{+})$), $\rho_{\alpha}$ could equally contain both regular and singular components, and so the negativity of $C_{\alpha}(t)$ is not necessarily caused by $\rho_{\alpha}(s)$ becoming negative as a continuous function. In fact, by inserting a singular term of the form $B\delta_{b}'$ into Eq.~(\ref{schwinger}) (with $b>0$ and $B<0$), it is clear that the appearance of such a term can cause $C_{\alpha}(t)$ to become continuously negative. Moreover, the shape characteristics of $C_{\alpha}(t)$ can also be replicated. Lattice calculations of both the quark and gluon propagators indicate that $C_{\alpha}(t)$ starts positive for small $t$, becomes negative at some specific value of $t$, and then generally flattens towards zero for large values of $t$. By employing a spectral density ansatz of the form: $\rho_{\alpha}(s)= A\delta(s-a) + B\delta'(s-b)$, where $A>0$ and $B<0$ ($a,b>0$ by definition), this qualitative behaviour is reproduced. This specific example demonstrates that the spectral density $\rho_{\alpha}$ can in fact be completely singular, and yet still reproduce the observed behaviour of $C_{\alpha}(t)$. In light of Corollary~\ref{qcd_spectral}, the characteristics of $C_{\alpha}(t)$ can therefore instead be interpreted as evidence of confinement caused by the failure of the CDP for quarks and gluons. \\

\noindent
In some cases, the spectral densities of propagators (and hence correlators) are explicitly computed using a combination of numerical and analytic approaches\footnote{See for example~\cite{Strauss_Fischer_Kellermann12} and~\cite{Dudal_Oliveira_Silva14}.}. This then means that the explicit form of $\widehat{\mathcal{T}}_{(1,2)}^{T}$ can be determined using Eq.~(\ref{KL_gen_rep}), and thus Corollary~\ref{N_corr} can actually be applied directly to determine whether $N=0$ or not. Of course, if one can demonstrate that $N \neq 0$, it is also interesting to establish the specific value that $N$ takes, especially in the case of the quark and gluon correlators where $N$ determines how the correlation strength $F(r)$ between the quark and gluon field clusters behaves as $r \rightarrow \infty$. However, as discussed in Sec.~\ref{paper4_section1}, if one constructs a bound on $\widehat{\mathcal{T}}_{(1,2)}^{T}$ (as in Theorem~\ref{bros_epstein_glaser}), one is not guaranteed that the largest value of $|\alpha|$ which appears in this bound is minimal. If the maximal value of $|\alpha|$ equals $m$ say, then one can only assert that the order of $\widehat{\mathcal{T}}_{(1,2)}^{T}$ satisfies: $N \leq m$. Nevertheless, this is still interesting because it implies an upper bound on the asymptotic behaviour of $F^{q\bar{q}}(r)$ and $F^{gg}(r)$. \\

\noindent
The discussions in this section demonstrate that confinement may occur in QCD because of a violation of the CDP caused by the appearance of non-measure defining $\delta_{a}'$ terms in the spectral densities of coloured correlators. In particular, the negativity of $C_{\alpha}(t)$ for quark and gluon propagators can be interpreted as evidence of such a violation, and therefore supports the hypothesis of quark and gluon confinement. Due to Theorem~\ref{cdt} and Corollary~\ref{N_corr}, whether a truncated cluster correlator violates the CDP or not depends on if its Fourier transform defines a measure, which itself is determined by the structure of the corresponding spectral densities $\{\rho_{\alpha} \}$. Although we will not pursue this further here, it would be interesting to establish whether the constraints imposed on $\{\rho_{\alpha} \}$ by the requirement to preserve or violate the CDP are consistent with the constraints due to other physical or structural relations, such as the operator product expansion~\cite{Lowdon15}.

\section{Conclusions}
\label{paper4_concl}

Determining the space-like asymptotic behaviour of truncated cluster correlators is crucial for understanding the large-distance correlation strength between the effective quanta associated with the field degrees of freedom in these correlators. For a QFT which is locally quantised, the cluster decomposition theorem implies that it is possible for a truncated cluster correlator to grow asymptotically, and hence violate the CDP, provided that the order of the correlator $N$ is non-vanishing. This possibility is particularly interesting in the context of QCD, as it provides a mechanism for which coloured degrees of freedom, such as quarks and gluons, can become confined. In this paper, we established a necessary and sufficient condition for a truncated cluster correlator to have $N=0$. It turns out that whether $N$ vanishes or not depends entirely on if the Fourier transform of the correlator defines a measure, which itself is determined by the structure of the spectral densities $\{\rho_{\alpha} \}$ of the correlator. Applying these results to QCD, it follows that if the indefinite inner product state space $\mathcal{V}_{\text{QCD}}$ has no mass gap, and any of the spectral densities of correlators involving coloured fields contains a $\delta_{a}'$ term, then this is sufficient to ensure confinement. In this context, the negativity of the Schwinger functions $C_{\alpha}(t)$ for the quark and gluon propagators on the lattice can therefore be interpreted as evidence that quarks and gluons are confined because they violate the CDP.

\section*{Acknowledgements}
I thank Thomas Gehrmann for useful discussions and input. This work was supported by the Swiss National Science Foundation (SNF) under contract CRSII2\_141847.

\renewcommand*{\cite}{\vspace*{-12mm}}


\begin{thebibliography}{99}

\bibitem{Strocchi76} F. Strocchi, ``Locality, charges and quark confinement,'' \textit{Phys. Lett. B} {\bf{62}}, 60 (1976).

\bibitem{Strocchi78} F. Strocchi, ``Local and covariant gauge quantum theories. Cluster property, superselection rules, and the infrared problem,'' \textit{Phys. Rev. D} {\bf{17}}, 2010 (1978).

\bibitem{Haag58} R. Haag, ``Quantum Field Theories with Composite Particles and Asymptotic Conditions,'' \textit{Phys. Rev.} {\bf{112}}, 669 (1958).

\bibitem{Araki60} H. Araki, ``On Asymptotic Behavior of Vacuum Expectation Values at Large Space-like Separation,'' \textit{Ann. of Phys.} {\bf{11}}, 260 (1960).

\bibitem{Ruelle62} D. Ruelle, ``On the Asymptotic Condition in Quantum Field Theory,'' \textit{Helv. Phys. Acta} {\bf{35}}, 147 (1962).

\bibitem{Araki_Hepp_Ruelle62} H. Araki, K. Hepp and D. Ruelle, ``On the Asymptotic Behaviour of Wightman Functions in Space-like Directions,'' \textit{Helv. Phys. Acta} {\bf{35}}, 164 (1962).

\bibitem{Streater_Wightman64} R. F. Streater and A. S. Wightman, \textit{PCT, Spin and Statistics, and all that}, W. A. Benjamin, Inc. (1964).

\bibitem{Nakanishi_Ojima90} N. Nakanishi and I. Ojima, \textit{Covariant Operator Formalism of Gauge Theories and Quantum Gravity}, World Scientific Publishing Co. Pte. Ltd (1990).

\bibitem{Haag96} R. Haag, \textit{Local Quantum Physics}, Springer-Verlag (1996). 

\bibitem{Strocchi13} F. Strocchi, \textit{An Introduction to Non-Perturbative Foundations of Quantum Field Theory}, Oxford University Press (2013). 

\bibitem{Roberts_Williams_Krein91} C. D. Roberts, A. G. Williams and G. Krein, ``On the Implications of Confinement,'' \textit{Int. J. Mod. Phys. A} {\bf{7}}, 5607 (1992).

\bibitem{Strichartz94} R. S. Strichartz, \textit{A Guide to Distribution Theory and Fourier Transforms}, CRC Press, Inc. (1994). 

\bibitem{Bros_Epstein_Glaser67} J. Bros, H. Epstein and V. Glaser, ``On the Connection Between Analyticity and Lorentz Covariance of Wightman Functions,''                  \textit{Comm. Math. Phys.} {\bf{6}}, 77 (1967).

\bibitem{Zhelobenko06} D. Zhelobenko, \textit{Principal Structures and Methods of Representation Theory}, American Mathematical Society (2006). 

\bibitem{Al-Gwaiz92} M. A. Al-Gwaiz, \textit{Theory of Distributions}, CRC Press, Inc. (1992).  

\bibitem{Rudin70} W. Rudin, \textit{Real and Complex Analysis}, McGraw-Hill Inc. (1970).

\bibitem{Baake_Grimm13} M. Baake and U. Grimm, \textit{Aperiodic Order: Volume 1, A Mathematical Invitation}, Cambridge University Press (2013). 

\bibitem{Bogolubov_Logunov_Oksak90} N. N. Bogolubov, A. A. Logunov and A. I. Oksak, \textit{General Principles of Quantum Field Theory}, Kluwer Academic Publishers (1990).

\bibitem{Cornwall13} J. M. Cornwall, ``Positivity violations in QCD,'' \textit{Mod. Phys. Lett. A} \textbf{28}, 1330035 (2013). 

\bibitem{Alkofer_Detmold_Fischer_Maris04} R. Alkofer, W. Detmold, C. S. Fischer and P. Maris, ``Analytic properties of the Landau gauge gluon and quark propagators,'' \textit{Phys. Rev. D} {\bf{70}}, 014014 (2004).  

\bibitem{Silva_Oliveira_Dudal_Bicudo_Cardoso13} P. Silva, O. Oliveira, D. Dudal, P. Bicudo and N. Cardoso, ``Many faces of the Landau gauge gluon propagator at zero and finite temperature: positivity violation spectral density and mass scales,'' \textit{Proceedings, QCD-TNT-III, From Quarks and Gluons to Hadronic Matter: A Bridge too Far?}, 040 (2013). 

\bibitem{Strauss_Fischer_Kellermann12} S. Strauss, C. S. Fischer and C. Kellermann, ``Analytic Structure of the Landau-Gauge Gluon Propagator,'' \textit{Phys. Rev. Lett.} \textbf{109}, 252001 (2012).

\bibitem{Dudal_Oliveira_Silva14} D. Dudal, O. Oliveira and P. J. Silva, ``K\"{a}ll\'{e}n-Lehmann spectroscopy for (un)physical degrees of freedom,'' \textit{Phys. Rev. D} \textbf{89}, 014010 (2014).

\bibitem{Lowdon15} P. Lowdon, ``Spectral density constraints in quantum field theory,'' \textit{Phys. Rev. D} {\bf{92}}, 045023 (2015).


\end{thebibliography}
\end{document}